\documentclass[12pt]{article}

\usepackage{amssymb,amsmath,graphics,color,eepic,amsthm,tikz}
\def\bbblb#1{\setbox\@tempboxa\hbox{$#1[$}%
             \@tempdimb\wd\@tempboxa
             \copy\@tempboxa \kern -.85\@tempdimb
             \copy\@tempboxa \kern -.65\@tempdimb\box\@tempboxa}
\def\bbbrb#1{\setbox\@tempboxa\hbox{$#1]$}%
             \@tempdimb\wd\@tempboxa
             \copy\@tempboxa \kern -.65\@tempdimb
             \copy\@tempboxa \kern -.85\@tempdimb\box\@tempboxa}

\def\a{{\boldsymbol a}}
\def\b{{\boldsymbol b}}
\def\m{{\boldsymbol m}}
\def\c{{\boldsymbol c}}

\def\q{{\boldsymbol q}}
\def\e{{\boldsymbol e}}
\def\h{{\boldsymbol h}}
\def\B{{\boldsymbol B}}

\def\p{{\boldsymbol p}}
\def\s{{\boldsymbol s}}
\def\T{{\boldsymbol T}}

\def\bPsi{{\boldsymbol \Psi}}
\def\bPhi{{\boldsymbol \Phi}}

\def\q{{\boldsymbol q}}

\def\bbbr{{\Bbb R}}

\def\bbbc{{\Bbb C}}

\def\bbbz{{\Bbb Z}}
\def\bbbd{{\Bbb D}}
\def\wedgecomma{\mathop{\wedge}\limits_{'}}
\def\ad{\mbox{ad}\,}

\def\rank{\mbox{rank}\,}
\def\tr{\mbox{tr}\,}

\def\im{{\rm Im}\,}

\def\id{\mbox{Id}\,}

\def\fr#1{{\mathfrak{#1}}}
\def\openone{\leavevmode\hbox{\small1\kern-3.3pt\normalsize1}}

\def\biglb{\big[\hspace*{-1.mm}\big[}
\def\bigrb{\big]\hspace*{-1.mm}\big]}

\def\bPsi{{\boldsymbol \Psi}}

\def\bPhi{{\boldsymbol \Phi}}

\newtheorem{theorem}{{Theorem}}[section]

\newtheorem{lemma}{Lemma}[section]

\numberwithin{equation}{section}
\providecommand{\keywords}[1]{\textbf{\textit{Keywords --}} #1}

\begin{document}

\title{ Kulish-Sklyanin type models: integrability and reductions}

\author{V. S.  Gerdjikov$^{1,2}$\footnote{Associated member of the Institute for Nuclear Research and Nuclear Energy, Bulgarian
Academy of Sciences, 72 Tzarigradsko Chaussee Blvd., 1784 Sofia, Bulgaria }\\
  \\[5pt]
{\sl $^1$ Institute of Mathematics and Informatics } \\ {\sl Bulgarian Academy of Sciences, }\\
{\sl Acad. Georgi Bonchev Str., Block 8, 1113 Sofia, Bulgaria}\\[5pt]
{\sl  $^2$ Institute for Advanced Physical Studies} \\
{\sl  New Bulgarian University,}\\
{\sl 21 Montevideo Street, Sofia 1618, Bulgaria}}

\date{}      % Deleting this command produces today's date.
\maketitle

\begin{abstract}
We start with a Riemann-Hilbert problem (RHP)  related to a {\bf BD.I}-type symmetric spaces $SO(2r+1)/S(O(2r -2s +1)\otimes O(2s))$,
$s\geq 1$. We consider two Riemann-Hilbert problems: the first formulated on the real axis $\mathbb{R}$ in the complex $\lambda$-plane; the second one
is formulated on  $\mathbb{R} \oplus i\mathbb{R}$. The first  RHP for $s=1$ allows one to solve the Kulish-Sklyanin (KS)  model; the second RHP
is relevant for a new type of KS model. An important example for nontrivial deep reductions of KS model is given. Its effect on the scattering matrix
is formulated. In particular we obtain  new 2-component NLS equations. Finally, using the Wronskian relations we demonstrate that the inverse scattering method for
KS models may be understood as a generalized Fourier transforms. Thus we have a tool to derive all their fundamental properties, including the hierarchy of
equations and he hierarchy of  their Hamiltonian structures.

\end{abstract}

\section{Introduction}

The Kulish-Sklyanin model \cite{KuSkl} may be viewed as an alternative version of the famous Manakov model \cite{ma74}.
Both models have important applications in physics, of which we will briefly mention the Bose-Einstein condensates, see
\cite{IMW04,LLMML05,CYHo,OM,Ho}. Both models are particular cases of the family of multi-component NLS (MNLS) equations,
which are related to the symmetric spaces \cite{ForKu*83}. They are integrable in the sense that they possess Lax representations
whose Lax operator is linear in $\lambda$ with potential $Q(x,t)$ taking value in simple Lie algebra (generalized Zakharov-Shabat system).
The solutions for the direct and the inverse scattering problem for such operators by now are  well known: see
\cite{ZMNP,FaTa,AblPrinTru*04,VSG2,ContMat,GGK05a,VRG-WMo,GeYaV}. The solutions of these equations can be derived effectively using the
dressing Zakharov-Shabat method \cite{ZaSha,Za*Mi,I04}, see also \cite{VSG2,1,GGK05a}.

Another important trend in soliton theory is based on the derivation of new MNLS equations by applying Mikhailov reductions
\cite{Mikh} to generic MNLS, see  \cite{1,I04,GGK05a,GGK05b,VG-RomJP}. In particular, recently we have been able to derive
two new 2-component MNLS \cite{Pliska16}; below we find a third example of such equation which follows from a generic KSm.

Most of the techniques mentioned above can be applied also to Lax operators which are quadratic in $\lambda$. Simplest cases of
such Lax operators related to the algebra $sl(2)$ have been considered before \cite{GIK-TMF44,GI-BJP10a,GI-BJP10b}. Recently an
alternative approach based on the Riemann-Hilbert problem allowed one to generalize these results to any simple Lie algebra
\cite{Pliska12,Pliska16}.
Therefore we will start with a RHP and then will demonstrate that: i) it is equivalent to the traditional
approach via the Lax representation; ii) it is a convenient tool to derive more general Lax representations
in which both operators are polynomial in $\lambda$ of order $k\geq 2$;  iii) allows one to analyze the
spectral properties not only of the Lax operators, but also of the relevant recursion operators. In conclusion
we will demonstrate that the RHP can be viewed as most effective approach to the fundamental properties
of the NLEE and especially to their hierarchies of Hamiltonian structures.

The family of Riemann-Hilbert problems (RHP) \cite{Pliska12} is of the form:
\begin{equation}\label{eq:rhp}\begin{aligned}
 \xi^+(x,t,\lambda) &=  \xi^-(x,t,\lambda) G(x,t,\lambda), &\qquad \lambda^k &\in \mathbb{R}, \\
 i \frac{\partial G}{ \partial x } &-\lambda^k [J, G(x,t,\lambda)] =0, &\qquad   i \frac{\partial G}{ \partial t } -
 \lambda^{2k} [J, G(x,t,\lambda)]&=0,
\end{aligned}\end{equation}
where $\xi^+(x,t,\lambda)$ take values in a simple Lie group $\mathcal{G}$ and $J$ is an element to
the corresponding simple Lie algebra $\mathfrak{g}$.
The specific choice of $k$, $\mathfrak{g}$ and $J$ determine the class and the type of NLEE under study.

The value of $k$ determines the order of the Lax operator $L$ as a polynomial in $\lambda$. The majority
of Lax operators $L(\lambda)$ that have been used to study most important NLEE are linear in $\lambda$, so $k=1$. There
have been also important examples of Lax operators that are quadratic in $\lambda$, i.e. $k=2$ \cite{GI-BJP10a,GI-BJP10b,GIK-TMF44}.
Note, that the order $p$ of the second Lax operators $M(\lambda)$ is always $p\geq k$.
Below we will stick to these two values of $k$ and will assume that $p=2k$, but the method outlined below can easily be extended to any
finite values of $k$ and $p$.
We will not touch here the other important classes of Lax operators that are rational functions of $\lambda$ \cite{Za*Mi}.

The choice of $\mathfrak{g}$ and $J\in \mathfrak{h}$ -- the Cartan subalgebra of $\mathfrak{g}$, determine the structure of the phase space of the NLEE;
this includes such important things as the number of the independent functions entering the NLEE and the Poisson brackets between
them. This has been well known for the generalized Zakharov-Shabat systems \cite{ZMNP,ZaSha}; the relevant Lax operators
 are linear in $\lambda$ ($k=1$) and take the form:
 \begin{equation}\label{eq:L}\begin{split}
 L\psi \equiv i \frac{\partial \psi}{ \partial x } + (Q(x,t) - \lambda J)\psi(x,t,\lambda) =0.
 \end{split}\end{equation}
The phase space $\mathcal{M}$ of the relevant NLEE or, in other words, the space of allowed potentials is defined as:
\begin{equation}\label{eq:M}\begin{split}
 \mathcal{M} \equiv \left\{ Q(x,t) = [J,X(x,t)], \qquad X(x,t) \in \mathfrak{g} \right\},
\end{split}\end{equation}
i.e. $Q(x,t)$ belongs to the co-adjoint orbit of $\mathcal{O}_J \in \mathfrak{g}$ passing through $J$.

We note here that if $\rank \mathfrak{g}>1$ and we choose $J$ with different real eigenvalues then we will
have NLEE on a homogeneous spaces. As typical representatives here we can mention the $N$-wave equations \cite{ZMNP}.
For other, more special choices of $J$ and $\mathfrak{g}$ our construction is on symmetric space.
These facts have been well known for very long time \cite{ma74,ForKu*83,KuSkl}. They can also be generalized for $k\geq 2$.

In Section 2 below we give some preliminaries concerning the structure of the BD.I symmetric spaces, RHP and
related Lax pairs and integrable equations, as well as the reduction group \cite{Mikh}.
In Section 3 we provide the construction of the Lax representations via the RHP method for $k=1$ and $k=2$.
The next Section 4 we give example of  $\mathbb{Z}_6$-reduction of a generic KSm
 thus deriving new integrable 2-component
NLS equations, see \cite{Pliska16}. In Section 5 we formulate the consequences of these reductions on the scattering matrix
and scattering data. In the last Section 6 we formulate the integrability properties of the KS type models. These are based
on the Wronskian relations and the `squared` solutions of the Lax operators. We prove that they are complete set of functions in the
phase space $\mathcal{M}$ and derive the expansions of $Q(x,t)$ and its variation $\ad_J^{-1} \delta Q$ over
the squared solutions. These results allow us  to formulate the fundamental properties of the NLEE and their Hamiltonian hierarchies.

\section{Preliminaries}

\subsection{On BD.I-type symmetric spaces}

For our specific purposes we will choose the simple Lie group $\mathcal{G}\simeq SO(2r+1)$, its Lie algebra
$\mathfrak{g}\simeq so(2r+1)$. The orthogonality condition that we will use below is
\begin{equation}\label{eq:ort}\begin{split}
X \in so(2r+1) \qquad \mbox{iff} \qquad X+S_0X^T S_0^{-1} =0 , \qquad S_0 = \sum_{k=1}^{2r+1}(-1)^{k+1} E_{k,2r+2-k},
\end{split}\end{equation}
where $E_{kn}$ is a $2r+1 \times 2r+1$-matrix with $(E_{kn})_{pj}=\delta_{kp}\delta_{nj}$.
This choice ensures that the  Cartan subalgebra $ \mathfrak{h}$ consists of diagonal matrices.
The element  $J\in \mathfrak{h}$ is chosen as:
\begin{equation}\label{eq:J}\begin{split}
J = \sum_{k=1}^{s}H_{e_k} = \left(\begin{array}{ccc} \openone_s & 0 & 0 \\ 0 & 0 & 0 \\ 0 & 0 & -\openone_s  \end{array}\right).
\end{split}\end{equation}
We will assume that the reader is familiar with the theory of simple Lie groups and algebras, see \cite{Helg}.
The system of positive roots of $so(2r+1)$ is well known \cite{Helg}:
\[ \Delta^+ = \{ e_i \pm e_j, \quad 1\leq i <j \leq r, \qquad e_j , \quad 1\leq j \leq r \}. \]
Here we also mention that
using $J$ and the Cartan involution one can introduce a $\mathbb{Z}_2$-grading in $\mathfrak{g}$:
\begin{equation}\label{eq:Z2}\begin{split}
C_1 &= \exp (\pi i J)= \left(\begin{array}{ccc} -\openone_s & 0 & 0 \\ 0 & \openone_{2r-2s+1} & 0 \\ 0 & 0 & -\openone_s  \end{array}\right), \qquad
\mathfrak{g} = \mathfrak{g}^{(0)} \oplus  \mathfrak{g}^{(1)} , \\
 \mathfrak{g}^{(0)} &\equiv \{
X\in  \mathfrak{g}^{(0)} \colon C_1X C_1^{-1}=X \}, \qquad \mathfrak{g}^{(1)} \equiv \{ Y\in  \mathfrak{g}^{(1)} \colon C_1Y C_1^{-1}=-Y  \}.
\end{split}\end{equation}
The $\mathbb{Z}_2$-grading means that
\begin{equation}\label{eq:Z2a}\begin{split}
[X_1, X_2] \in \mathfrak{g}^{(0)} , \qquad [X_1, Y_1] \in \mathfrak{g}^{(1)} , \qquad [Y_1, Y_2] \in \mathfrak{g}^{(0)} ,
\end{split}\end{equation}
and provides the local structure of the symmetric space of BD.I-class $SO(2r+1) /(SO(2r-2s+1)\otimes SO(2s))$.
We will see that this construction is directly related to the KSm for $s=1$.
It is well known that the set of positive roots and the  system of simple roots of $B_r$ are:
\begin{equation}\label{eq:B-r}\begin{split}
\Delta_{B_r}^+ &\equiv \{ e_i - e_j , \quad e_i+e_j , \quad 1 \leq i < j \leq r; \qquad e_j, \quad 1\leq j \leq r \} \\
\pi_{B_r} &\equiv \{ \alpha_k = e_k - e_{k+1}, \quad \alpha_r= e_r, \quad  1\leq k \leq r-1.\}
\end{split}\end{equation}
Introducing $J$ as in (\ref{eq:J}) we split the system of positive roots into two subsets:
\begin{equation}\label{eq:Jroot}\begin{split}
\Delta_1^+ = \{\beta, \quad  \beta(J) = 1\} \qquad \Delta_0^+ = \{\beta, \quad  \beta(J) = 0 \mod 2\}.  .
\end{split}\end{equation}
Below we will consider two cases with $s=1$ and $s=3$:
\begin{equation}\label{eq:Delta}\begin{aligned}
& s=1, r\geq 3 &\quad &\begin{aligned} \Delta_0^+ &= \{  e_i \pm  e_j \quad 2\leq i < j \leq r, \quad e_j \quad 2\leq i  \leq r \}  , \\
\Delta_1^+ &= \{ e_1 \pm e_j ,  \quad 2\leq j  \leq r , \quad e_1 \};
\end{aligned} \\
& s=3, r=4 &\quad &\begin{aligned} \Delta_0^+ &= \{  e_1 \pm  e_2 , e_1 \pm  e_3 , e_2 \pm  e_3 , e_1, e_2 , e_3 \}  , \\
\Delta_1^+ &= \{ e_1 \pm e_4 ,  e_2 \pm e_4 ,  e_3 \pm e_4 ,   e_4 \}; \end{aligned}
\end{aligned}\end{equation}
The potential of the Lax operator as well as the coefficients $Q_{2s-1}$ are given by:
\begin{equation}\label{eq:Q}\begin{split}
 Q(x,t) &= \sum_{\alpha \in \Delta_1^+}^{} (q_\alpha E_{\alpha} + p_\alpha E_{-\alpha} ) \in \mathfrak{g}^{(1)}.
\end{split}\end{equation}

\subsection{The RHP and Lax representations}

We will start by formulating the RHP (\ref{eq:rhp}). Given the sewing function $G(x,t,\lambda)$ find
the functions $\xi^\pm(x,t,\lambda)$ taking values in the simple Lie group $\mathcal{G}$ and analytic for
$\im \lambda^k \gtrless 0$ such that eq. (\ref{eq:rhp}) holds.
It is natural to impose also the normalization condition
.\begin{equation}\label{eq:norm}\begin{split}
\lim_{\lambda \to \infty} \xi^\pm(x,t,\lambda) = \openone,
\end{split}\end{equation}
which ensures that the RHP has unique regular solution, on Figure 1 we .show the analyticity regions and
the contours that will be used below for $k=1$ and $k=2$.

\begin{theorem}[Zakharov-Shabat theorem]\label{thm:}
Let $\xi^\pm (x,t,\lambda) $ be a solution of the RHP  whose sewing function satisfies the equations
in (\ref{eq:rhp}). Then $\xi^\pm (x,t,\lambda) $ is a fundamental analytic solution (FAS) of the operators:
\begin{equation}\label{eq:lm}\begin{aligned}
L \xi^\pm  &\equiv i \frac{\partial \xi^\pm }{ \partial x } + U(x,t,\lambda) \xi^\pm (x,t,\lambda) - \lambda^k [J,\xi^\pm (x,t,\lambda)]=0, \\
M \xi^\pm  &\equiv i \frac{\partial \xi^\pm }{ \partial x } + V(x,t,\lambda) \xi^\pm (x,t,\lambda) - \lambda^{2k} [J,\xi^\pm (x,t,\lambda)]=0, \\
\end{aligned}\end{equation}
where $U(x,t,\lambda)$ and $V(x,t,\lambda)$ are polynomials in $\lambda$ of order $k-1$ and $2k-1$ respectively:
\begin{equation}\label{eq:UV}\begin{split}
U(x,t,\lambda) = \lambda^k J - (\lambda^k \xi J\hat{\xi}(x,t,\lambda))_+, \qquad V(x,t,\lambda) =
\lambda^{2k} J - (\lambda^2k \xi J\hat{\xi}(x,t,\lambda))_+,
\end{split}\end{equation}

\end{theorem}

\begin{proof}[Idea of the proof.]
 Consider the functions
 \begin{equation}\label{eq:gpm}\begin{split}
g^\pm (x,t,\lambda) &= i \frac{\partial \xi^\pm}{ \partial x } (\xi^\pm)^{-1} (x,t,\lambda) + \lambda^k \xi^\pm (x,t,\lambda) J (\xi^\pm)^{-1} (x,t,\lambda), \\
f^\pm (x,t,\lambda) &= i \frac{\partial \xi^\pm}{ \partial t } (\xi^\pm)^{-1} (x,t,\lambda) + \lambda^{2k} \xi^\pm (x,t,\lambda) J (\xi^\pm)^{-1} (x,t,\lambda),
 \end{split}\end{equation}
 and using the explicit $x$ and $t$-dependence of $G(x,t,\lambda)$ prove that $g^+ (x,t,\lambda)=g^- (x,t,\lambda)$ and
 $f^+ (x,t,\lambda)=f^- (x,t,\lambda)$. Then, using eq. (\ref{eq:norm}) we find  that
 \begin{equation}\label{eq:gauge}\begin{split}
  \lim_{\lambda \to\infty} g^\pm (x,t,\lambda) = \lambda^k J, \qquad   \lim_{\lambda \to\infty} f^\pm (x,t,\lambda) = \lambda^{2k} J,
 \end{split}\end{equation}
It remains to apply the great Liouville theorem that ensures that the functions  $g^+ (x,t,\lambda)=g^- (x,t,\lambda)$
(resp.  $f^+ (x,t,\lambda)=f^- (x,t,\lambda)$) are analytic on the whole complex $\lambda^k$-plane and therefore are
polynomial in $\lambda$ of order $k$ (resp. $2k$). Of course the coefficients of these polynomials may depend on $x$ and $t$.
 The explicit relations (\ref{eq:UV}) of these coefficients with the solution $\xi^\pm (x,t,\lambda)$ has been
 proposed by  Drinfeld-Sokolov \cite{DriSok}, see also \cite{Pliska12}.

\end{proof}

The  explicit derivation of the Lax pairs is very effective if $\xi^\pm(x,t,\lambda)$ satisfy the canonical
normalization (\ref{eq:norm}). Indeed, in this case we can use the asymptotic expansion:
\begin{equation}\label{eq:xi}\begin{split}
\xi ^\pm(x,t,\lambda) &=   \exp \left( \mathcal{Q}(x,t,\lambda)\right), \qquad \mathcal{Q}(x,t,\lambda)
 =\sum_{s=1}^{\infty} \lambda^{-2s+1} Q_{2s-1} .
\end{split}\end{equation}
Obviously such choice for $\xi^\pm (x,t,\lambda)$ involves a $\mathbb{Z}_2$ reduction, which can be
formulated in several equivalent ways, e.g.:
\begin{equation}\label{eq:xipm}\begin{split}
& \mbox{a)} \qquad \xi ^+(x,t, -\lambda) = (\xi ^-)^{ -1}(x,t,\lambda), \qquad Q_{2s-1} \in \mathfrak{g}^{(1)},
\\  &\mbox{b)} \qquad
 \xi ^+(x,t, \lambda^*) = (\xi ^{-}) ^\dag (x,t,\lambda) \qquad \mbox{iff} \qquad Q_{2s-1} =-Q^\dag _{2s-1},
\end{split}\end{equation}
In order to derive the relevant NLEE in terms of $Q_{2s-1}$ we will use the formulae:
\begin{equation}\label{eq:xias}\begin{split}
\xi ^\pm J\hat{\xi}^\pm(x,t,\lambda) & = J + \sum_{p=1}^{\infty} \frac{1}{p!}\; \ad_{\mathcal{Q}}^p J. \qquad
\frac{\partial \xi^\pm }{ \partial x }\hat{\xi}^\pm(x,t,\lambda)= \frac{\partial \mathcal{Q}}{ \partial x }
 + \sum_{p=1}^{\infty} \frac{1}{(p+1)!}\; \ad_{\mathcal{Q}}^p \frac{\partial \mathcal{Q}}{ \partial x }.
\end{split}\end{equation}
which allow us to express the Lax pair coefficients $U_s(x,t)$ and $V_s(x,t)$ in terms of $Q_{2s-1} $ and their
$x$-derivatives. In all our considerations we will need only the first  few terms of these expansions;  for more
details see \cite{Pliska12}.

\subsection{The reduction group}
Following Mikhailov \cite{Mikh} we will also impose additional  reductions using the famous reduction group $G_R $.

$G_R$  is a finite group which preserves the Lax representation (\ref{eq:lm}), i.e. it ensures that the reduction
constraints are automatically compatible with the evolution. $G_R $ must have two realizations: i) $G_R \subset {\rm Aut}\fr{g} $ and ii) $G_R
\subset {\rm Conf}\, \mathbb{C} $, i.e. as conformal mappings of the complex $\lambda $-plane. To each $g_k\in G_R $ we relate a reduction
condition for the Lax pair as follows:
\begin{equation}\label{eq:2.1}
C_k(L(\Gamma _k(\lambda ))) = \eta _k L(\lambda ), \quad C_k(M(\Gamma _k(\lambda ))) = \eta _k M(\lambda ),
\end{equation}
where $C_k\in \mbox{Aut}\; \fr{g} $ and $\Gamma _k(\lambda )\in \mbox{Conf\,} \bbbc $ are the images of $g_k $ and $\eta _k =1 $ or $-1 $
depending on the choice of $C_k $. Since $G_R $ is a finite group then for each $g_k $ there exist an integer $N_k $ such that $g_k^{N_k} =\openone$.
More specifically, below we will consider $\mathbb{Z}_2$-reductions of the form:
\begin{align}\label{eq:U-V.a}
&\mbox{a)} &\quad
 B_1U^{\dagger}(\kappa _1(\lambda ))B_1^{-1} &= U(\lambda ), &\quad B_1(V^{\dagger}(\kappa _1(\lambda ))B_1^{-1} &= V(\lambda ), \\
\label{eq:U-V.b}
&\mbox{b)} &\quad
 B_2U^{T}(\kappa _2(\lambda ))B_2^{-1} &= -U(\lambda ), &\quad B_2(V^{T}(\kappa _2(\lambda ))B_2^{-1} &= -V(\lambda ), \\
\label{eq:U-V.c}
&\mbox{c)} &\quad
 B_3U^{*}(\kappa _1(\lambda ))B_3^{-1} &= -U(\lambda ), &\quad  B_3(V^{*}(\kappa _1(\lambda ))B_3^{-1} &= -V(\lambda ), \\
\label{eq:U-V.d}
&\mbox{d)} &\quad
 B_4U(\kappa _2(\lambda ))B_4^{-1} &= U(\lambda ), &\quad  B_4(V(\kappa _2(\lambda ))B_4^{-1} &= V(\lambda ),
\end{align}
where the automorphisms $B_k$ must of finite order. In the cases (\ref{eq:U-V.a}), (\ref{eq:U-V.b}) and
(\ref{eq:U-V.c}) $B_k$ must be of even order, which in general could be bigger than 2.

Beside the  $\mathbb{Z}_2$-reductions we will impose additional $\mathbb{Z}_p$-reductions with $p >2$:
\begin{equation}\label{eq:Zp}\begin{split}
A_p U(x,t, \kappa_p(\lambda)) A_p^{-1} &= U(x,t, \lambda), \qquad A_p V^*(x,t, \kappa_p(\lambda)) A_p^{-1} = V(x,t,\lambda), \\
A_p (\xi ^+) (x,t, \kappa_p(\lambda)) A_p^{-1} &= (\xi ^+)^{-1} (x,t, \lambda), \qquad \kappa_p(\lambda) =\lambda \omega, \qquad
\omega = e^{2\pi i /p},
\end{split}\end{equation}
where  $A_p$ is an automorphism of $\mathfrak{g}$ of order $p$. Typically we will use a realization of
$A_p$ as an element of the Weyl group of $\mathfrak{g}$.

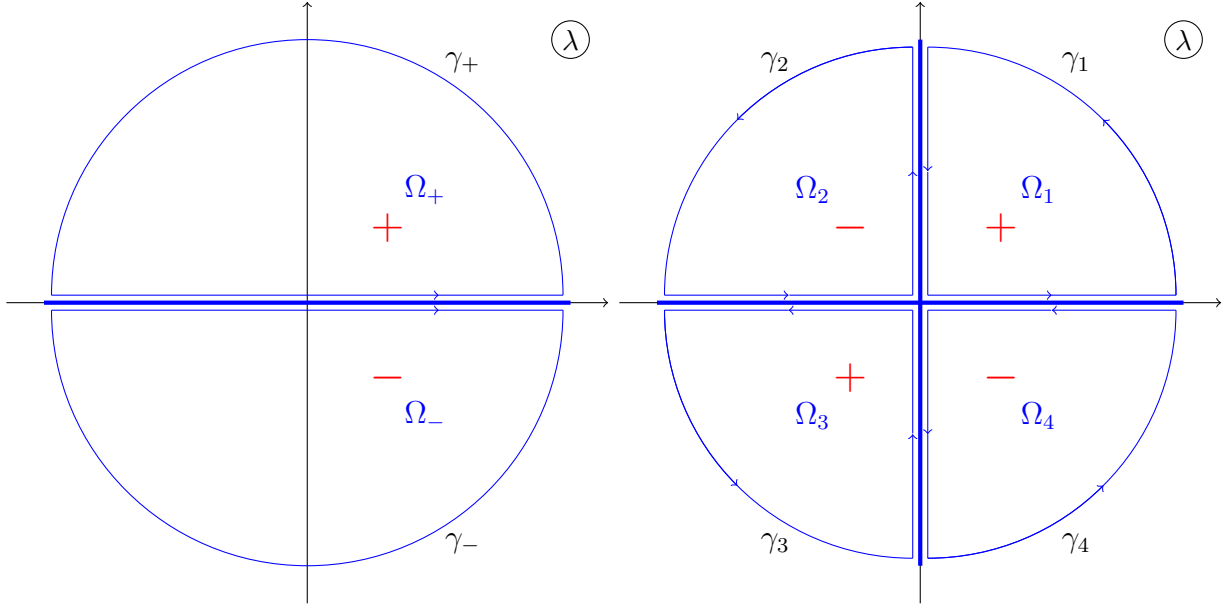
\begin{figure}
%\centerline{\includegraphics[width=0.89\textwidth]{spectra-sl3.eps}}
\begin{center}
\begin{tikzpicture}

\draw [->,black](0,-4) -- (0,4);
\draw [->,black](-4,0) -- (4,0);
\draw [-,ultra thick, blue](-3.5,0) -- (3.5,0);
%\draw [-,ultra thick, blue](0,-3.5) -- (0,3.5);
\draw [black] (3.5,3.5) node  {{\small $\lambda$}};
\draw [black] (3.5,3.5) circle (0.25);
\draw [black] (2,3.2) node  {{  $\gamma_+$}};
%\draw [black] (-2,3.2) node  {{  $\gamma_2$}};
%\draw [black] (-2,-3.2) node  {{  $\gamma_3$}};
\draw [black] (2,-3.2) node  {{  $\gamma_-$}};

\draw [red] (1.0,1.0) node  {{  \Large $+$}};
\draw [blue] (1.5,1.5) node  {{  $\Omega_+$}};
%\draw [red] (-1.0,1.0) node  {{ \Large $-$}};
%\draw [blue] (-1.5,1.5) node  {{  $\Omega_2$}};
%\draw [blue] (-1.5,-1.5) node  {{  $\Omega_3$}};
%\draw [red] (-1.0,-1.0) node  {{ \Large $+$}};
\draw [blue] (1.5,-1.5) node  {{  $\Omega_-$}};
\draw [red] (1.0,-1.0) node  {{ \Large $-$}};

\draw [-, blue](3.4,0.1) arc (0:180:3.4);
\draw [-, blue,rotate=180](3.4,0.1) arc (0:180:3.4);
%\draw [->, blue](3.4,0.1) arc (0:45:3.3);
\draw [->,blue](-3.4,0.1) -- (1.75,0.1);
\draw [-,blue](1.75,0.1) -- (3.4,0.1);
\draw [->,blue](-3.4,-0.1) -- (1.75,-0.1);
\draw [-,blue](1.75,-0.1) -- (3.4,-0.1);

\end{tikzpicture}
\begin{tikzpicture}

\draw [->,black](0,-4) -- (0,4);
\draw [->,black](-4,0) -- (4,0);
\draw [-,ultra thick, blue](-3.5,0) -- (3.5,0);
\draw [-,ultra thick, blue](0,-3.5) -- (0,3.5);
\draw [black] (3.5,3.5) node  {{\small $\lambda$}};
\draw [black] (3.5,3.5) circle (0.25);
\draw [black] (2,3.2) node  {{  $\gamma_1$}};
\draw [black] (-2,3.2) node  {{  $\gamma_2$}};
\draw [black] (-2,-3.2) node  {{  $\gamma_3$}};
\draw [black] (2,-3.2) node  {{  $\gamma_4$}};

\draw [red] (1.0,1.0) node  {{  \Large $+$}};
\draw [blue] (1.5,1.5) node  {{  $\Omega_1$}};
\draw [red] (-1.0,1.0) node  {{ \Large $-$}};
\draw [blue] (-1.5,1.5) node  {{  $\Omega_2$}};
\draw [blue] (-1.5,-1.5) node  {{  $\Omega_3$}};
\draw [red] (-1.0,-1.0) node  {{ \Large $+$}};
\draw [blue] (1.5,-1.5) node  {{  $\Omega_4$}};
\draw [red] (1.0,-1.0) node  {{ \Large $-$}};

\draw [-, blue](3.4,0.1) arc (0:90:3.3);
\draw [->, blue](3.4,0.1) arc (0:45:3.3);
\draw [->,blue](0.1,0.1) -- (1.75,0.1);
\draw [-,blue](1.75,0.1) -- (3.4,0.1);
\draw [->,blue](0.1,3.4) -- (0.1,1.75);
\draw [-,blue](0.1,1.74) -- (0.1,0.1);

\draw [-, blue,rotate=90](3.4,0.1) arc (0:90:3.3);
\draw [->, blue,rotate=90](3.4,0.1) arc (0:45:3.3);
\draw [->,blue,rotate=90](0.1,0.1) -- (1.75,0.1);
\draw [-,blue,rotate=90](1.75,0.1) -- (3.4,0.1);
\draw [->,blue,rotate=90](0.1,3.4) -- (0.1,1.75);
\draw [-,blue,rotate=90](0.1,1.74) -- (0.1,0.1);

\draw [-, blue,rotate=180](3.4,0.1) arc (0:90:3.3);
\draw [->, blue,rotate=180](3.4,0.1) arc (0:45:3.3);
\draw [->,blue,rotate=180](0.1,0.1) -- (1.75,0.1);
\draw [-,blue,rotate=180](1.75,0.1) -- (3.4,0.1);
\draw [->,blue,rotate=180](0.1,3.4) -- (0.1,1.75);
\draw [-,blue,rotate=180](0.1,1.74) -- (0.1,0.1);

\draw [-, blue,rotate=270](3.4,0.1) arc (0:90:3.3);
\draw [->, blue,rotate=270](3.4,0.1) arc (0:45:3.3);
\draw [->,blue,rotate=270](0.1,0.1) -- (1.75,0.1);
\draw [-,blue,rotate=270](1.75,0.1) -- (3.4,0.1);
\draw [->,blue,rotate=270](0.1,3.4) -- (0.1,1.75);
\draw [-,blue,rotate=270](0.1,1.74) -- (0.1,0.1);

\end{tikzpicture}
\end{center}
\caption{The continuous spectrum of a $L(\lambda)$ in thick blue, the analyticity regions $\Omega_s$ and the contours $\gamma_s$, $s=1,\dots, k$. Left
panel: $k=1$; $L(\lambda)$ is linear in $\lambda$ (\ref{eq:L}). Right panel: $k=2$: $L(\lambda)$ is quadratic in $\lambda$ (\ref{eq:UV2}).  }
\label{fig:1}
\end{figure}

\section{ Generic Kulish-Sklyanin type models }

The sewing function of our RHP problem depends on two additional parameters $x$ and $t$ in a special
way, which makes it convenient to analyze and solve special classes of nonlinear evolution equations (NLEE)
in two-dimensional space-time, known also as soliton equations. Here we also assume that the relevant
symmetric space is $SO(2r+1)/(SO(2r-1)\times SO(2))$.

\subsection{The compatibility condition $k=1$, $s=1$. }

In the simplest case  $k=1$ and $s=1$ the Lax pair:

\begin{equation}\label{eq:UV1}\begin{aligned}
U^{(1)}(x,t,\lambda) &= U_1(x,t)  -\lambda J, \qquad V^{(1)}(x,t,\lambda) = i Q_{1,x}+ V_2(x,t) +\lambda V_1(x,t) -\lambda^2 J,
\end{aligned}\end{equation}
where
\begin{equation}\label{eq:UV2}\begin{aligned}
U_1(x,t) &= [J,Q_1(x,t)] = Q(x,t) , &\quad V_1(x,t) &= Q(x,t) &\quad  V_2(x,t) &= \frac{1}{2}\ad_{Q_1} Q(x,t) ,
\end{aligned}\end{equation}
\begin{equation}\label{eq:UV4}\begin{split}
Q_1(x,t) = \left(\begin{array}{ccc} 0 & \vec{q}^T & 0 \\ -\vec{p} & 0 & s_0 \vec{q} \\
0 & -\vec{p}^Ts_0 & 0 \end{array}\right), \qquad   V_2(x,t) = \left(\begin{array}{ccc} (\vec{q},\vec{p}) & 0 & 0 \\
0 & s_0 \vec{q} \vec{p}^T s_0 - \vec{p}\vec{q}^T & 0 \\ 0 &  0 & -(\vec{q},\vec{p}) \end{array}\right).
\end{split}\end{equation}
As a result we get that this Lax pair leads to the well known Kulish-Sklyanin model whose integrability has been known since 1981
\cite{KuSkl}:
\begin{equation}\label{eq:KSm}\begin{split}
i \frac{\partial \vec{q}}{ \partial t} + \frac{\partial^2 \vec{q} }{ \partial x^2 } + 2 (\vec{q}\;^\dag, \vec{q})\vec{q}
- (\vec{q}^T s_0 \vec{q}) s_0 \vec{q}\;^* =0, \quad s_0 = \sum_{k=1}^{2r-1} (-1)^k E_{k,2r-k}.
\end{split}\end{equation}
where now $E_{kn}$ is a $2r-1 \times 2r-1$-matrix with $(E_{kn})_{pj}=\delta_{kp}\delta_{nj}$.
For applications of this model to Bose-Einstein condensates and detailed analysis for the inverse spectral
transform see \cite{VSG2}.

\subsection{The compatibility condition $k=1$, $s>1$. }

For $k=1$ and $s>1$ the Lax pair takes the form:
\begin{equation}\label{eq:UV1a}\begin{aligned}
U(x,t,\lambda) &= Q(x,t)  -\lambda J, \qquad V^{(1)}(x,t,\lambda) = iQ_{1,x}+ V_2(x,t) +\lambda Q(x,t) -\lambda^2 J,
\end{aligned}\end{equation}
where
\begin{equation}\label{eq:UV4'}\begin{split}
Q(x,t) = \left(\begin{array}{ccc} 0 & \q  & 0 \\ \p & 0 &  \tilde{\q} \\ 0 & \tilde{\p} & 0 \end{array}\right), \qquad
V_2(x,t) = \left(\begin{array}{ccc} - \q \p & 0 & 0 \\ 0 & \p \q - \tilde{\q} \tilde{\p} & 0 \\ 0 &  0 & \tilde{\q} \tilde{\p} \end{array}\right).
\end{split}\end{equation}
The matrix $S_0$ from orthogonality condition (\ref{eq:ort}) in this case equals $\left(\begin{array}{ccc}0 & 0 & \s_1 \\ 0 & \s_2 & 0 \\ \s_1^{-1} & 0 & 0
 \end{array}\right)$. The blocks $\s_k$, $k=1, 2$ are easily determined from (\ref{eq:ort}) and satisfy $s_2^2=\openone$, $s_2^{-1} =\s_2$. Then
$\tilde{\q}= -\s_2 \q^T \s_1$ and $\tilde{\p}= -\s_2 \p^T \s_1^{-1}$.
As a result we get the generic Kulish-Sklyanin model \cite{KuSkl}:
\begin{equation}\label{eq:KSm2}\begin{split}
i \frac{\partial \q }{ \partial t} + \frac{\partial^2 \q }{ \partial x^2 } + 2 \q \p \q - \q \tilde{\q} \tilde{\p} =0, \qquad
i \frac{\partial \p }{ \partial t} - \frac{\partial^2 \p }{ \partial x^2 } - 2 \p \q \p +  \tilde{\q} \tilde{\p} \p =0,
\end{split}\end{equation}
After the additional reduction $\p = \q^\dag$ we get:
\begin{equation}\label{eq:KSm1}\begin{split}
i \frac{\partial \q }{ \partial t} + \frac{\partial^2 \q }{ \partial x^2 } + 2 \q \q^\dag \q - \q \tilde{\q} \tilde{\q}^\dag =0.
\end{split}\end{equation}

\subsection{The compatibility condition $k=2$ }

Now $U^{(2)}(x,t,\lambda)$ is quadratic in $\lambda$, $V^{(4)}(x,t,\lambda)$ is quartic in $\lambda$ and we keep $s=1$:
\begin{equation}\label{eq:UV3}\begin{aligned}
U^{(2)}(x,t,\lambda) &= U_2(x,t)  + \lambda Q(x,t) -\lambda^2 J, \\
 V^{(4)}(x,t,\lambda) &=  V_4(x,t) +\lambda V_3(x,t) + \lambda^2 V_2(x,t) + \lambda^3 -\lambda^4 J,
\end{aligned}\end{equation}
where $Q(x,t)$ is given by (\ref{eq:UV4'}) and $U_2(x,t)=V_2(x,t)$ and in addition
\begin{equation}\begin{split}\label{eq:V3}
  V_3(x,t) = \left(\begin{array}{ccc} 0 & V_{3;2}^T & 0 \\ V_{3;1} & 0 & s_0V_{3;2} \\ 0 &  V^T_{3;1} s_0 & 0   \end{array}\right) ,
 \qquad \begin{split}
 V_{3;1} &= \vec{w} + \frac{1}{3} (\vec{p}^T s_0 \vec{p})  s_0\vec{q} - \frac{2}{3} (\vec{q},\vec{p}) \vec{p}, \\
 V_{3;2} &= s_0\vec{v} + \frac{1}{3} (\vec{q}^T s_0 \vec{q}) \vec{p} - \frac{2}{3} (\vec{q},\vec{p}) s_0\vec{q}.
\end{split}
\end{split}\end{equation}

\begin{equation}\label{eq:V4}\begin{split}
V_4&= \left(\begin{array}{ccc} V_{4;11} & 0 & 0 \\ 0 & V_{4;22}  & 0 \\ 0 & 0 & V_{4;33} \end{array}\right) , \qquad
Q_3(x,t) = \left(\begin{array}{ccc} 0 & \vec{v}^T & 0 \\ -\vec{w} & 0 & s_0 \vec{v} \\ 0 & -\vec{w}^Ts_0 & 0 \end{array}\right),\quad
\\ V_{4;11} &= -V_{4;33}
= i \left( (\vec{q}_x,\vec{p}) -  (\vec{q},\vec{p}_x)\right) + (\vec{q},\vec{p})^2 - \frac{1}{2}  (\vec{q}^Ts_0\vec{q}) (\vec{p}^Ts_0\vec{p}) , \\
V_{4;22} &=  i \left( \vec{p}_x \vec{q}^T -\vec{p} \vec{q}_x^T  + s_0( \vec{q}_x \vec{p}^T -\vec{q} \vec{p}_x^T) s_0\right) -
(\vec{q}.\vec{p}) \left( \vec{p} \vec{q}^T - s_0 \vec{q} \vec{p}^T s_0 \right).
\end{split}\end{equation}

The compatibility condition reads:
\begin{equation}\label{eq:LM2}\begin{aligned}
&\lambda^3: &\qquad i \frac{\partial Q}{ \partial x } +[U_2,Q]+[Q,V_2] &=[J, V_3], \\
&\lambda^2: &\qquad i \frac{\partial V_2}{ \partial x } +[U_2,V_2]+[Q,V_3] &=[J, V_4] ,
\end{aligned}\end{equation}
Since $U_2=V_2$ the first of the equations (\ref{eq:LM2}) gives:
\begin{equation}\label{eq:V3'}\begin{split}
V_3(x,t) = \ad_J^{-1} i \frac{\partial Q}{ \partial x }= i \left(\begin{array}{ccc} 0 & \vec{q}_x^T & 0 \\
-\vec{p}_x & 0 & s_0\vec{q}_x \\ 0 & - \vec{p}_x^T s_0 & 0  \end{array}\right) ,
\end{split}\end{equation}
which in components give:
\begin{equation}\label{eq:V3a}\begin{split}
\vec{v} &= i \frac{\partial \vec{q}}{ \partial x } + \frac{2}{3} (\vec{p},\vec{q})\vec{q} - \frac{1}{3} (\vec{q}^Ts_0\vec{q}) s_0\vec{p}, \qquad
\vec{w} =- i \frac{\partial \vec{p}}{ \partial x } + \frac{2}{3} (\vec{p},\vec{q})\vec{p} - \frac{1}{3} (\vec{p}^Ts_0\vec{p}) s_0\vec{q}.
\end{split}\end{equation}
The second equation in (\ref{eq:LM2}) is identically satisfied as a consequence of (\ref{eq:V3a}).

Finally the equations of motion:
\begin{equation}\label{eq:eq1}\begin{split}
\lambda: &\qquad i \frac{\partial V_3}{ \partial x }  - i \frac{\partial Q}{ \partial t } +[Q,V_4] +[U_2, V_3]=0, \qquad
\lambda^0: \qquad i \frac{\partial V_4}{ \partial x } - i \frac{\partial U_2}{ \partial t } +[U_2,V_4] =0.
\end{split}\end{equation}

Since in addition we put $\vec{p} =\epsilon \vec{q}\;^*$ we get:
\begin{equation}\label{eq:kusk1}\begin{split}
i \frac{\partial \vec{q}}{ \partial t} + \frac{\partial^2 \vec{q} }{ \partial x^2 }  + i\epsilon s_0 (\vec{q}s_0\vec{q})
 \frac{\partial \vec{q}\;^*}{ \partial x } + s_0 (\vec{q}, \vec{q}\;^*)  (\vec{q}s_0\vec{q}) s_0 \vec{q}\;^*
 - \left( \frac{1}{2} |(\vec{q}s_0\vec{q})|^2 + 2i \epsilon (\vec{q}, \vec{q}\;^*) \right) \vec{q} =0.
\end{split}\end{equation}
One can check that the second equation in (\ref{eq:eq1}) holds identically as a consequence of (\ref{eq:kusk1}).

\subsection{The gauge equivalent systems $k=2$}
 We fix up the gauge by requesting that $U(x,t,\lambda=0) =0$ and $V(x,t,\lambda=0) =0$. Thus
\begin{equation}\label{eq:lmt}\begin{split}
\tilde{L}\tilde{\psi}& \equiv i\frac{\partial \tilde{\psi}}{ \partial x } + (\lambda \tilde{Q}(x,t) - \lambda^2 J )\tilde{\psi}(x,t,\lambda)=0, \\
\tilde{M}\tilde{\psi}& \equiv i\frac{\partial \tilde{\psi}}{ \partial t } + (\lambda \tilde{V}_3(x,t) + \lambda^2 \tilde{V}_2(x,t)
+\lambda^3 \tilde{Q}(x,t)- \lambda^2 J )\tilde{\psi}(x,t,\lambda)=0,
\end{split}\end{equation}
Therefore $\tilde{\psi}(x,t,\lambda) = g_0(x,t) \psi (x,t,\lambda)$, i.e
\begin{equation}\label{eq:gauge0}\begin{split}
 i \frac{\partial g_0}{ \partial x } - g_0(x,t) U_2(x,t) =0, \quad i \frac{\partial g_0}{ \partial t } - g_0(x,t) V_4(x,t) =0,
\end{split}\end{equation}
so $g_0(x,t)$ is a block-diagonal function. Then  $\tilde{V}_k$ and $\tilde{Q}$  have
the same block structure of  $V_k$ and $Q$; we will denote their blocks by additional `tilde`.
Then  the NLEE get the form:
\begin{equation}\label{eq:nleet}\begin{split}
i \frac{\partial \vec{\tilde{q}}}{ \partial t} + \frac{\partial^2 \vec{\tilde{q}} }{ \partial x^2 } + i \frac{\partial }{ \partial x } \left(
2 \vec{\tilde{q}} (\vec{\tilde{q}\;}^T \vec{\tilde{p}}) - s_0 \vec{\tilde{p}} (\vec{\tilde{q}\;}^T \vec{\tilde{q}}) \right) =0 , \\
i \frac{\partial \vec{\tilde{p}}}{ \partial t} - \frac{\partial^2 \vec{\tilde{p}} }{ \partial x^2 } - i \frac{\partial }{ \partial x } \left(
2 \vec{\tilde{p}} (\vec{\tilde{q}\;}^T \vec{\tilde{p}}) - s_0 \vec{\tilde{q}} (\vec{\tilde{p}\;}^T \vec{\tilde{p}}) \right) =0,
\end{split}\end{equation}

\section{Generic KS type models  and their $\mathbb{Z}_n$-reductions}
Here we derive KS type models related to generic BD.I symmetric spaces $SO(2r+1)/(SO(2k)\times SO(2r-2k+1))$.
These are rather complicated systems of equations for $k(2r-2k+1)$ functions of $x$ and $t$. In the second
subsection we consider special case with $r=4$ and $k=3$ apply to it a special $\mathbb{Z}_6$-reduction.
The result is a new type of 2-component NLS.

\subsection{$\mathbb{Z}_6$-reduction of KS models for $SO(9)/(SO(6)\times SO(3))$}

Consider $SO(9)/(SO(6)\times SO(3))$. The the subset of positive roots is split into
\begin{equation}\label{eq:Delt0-1}\begin{aligned}
\Delta_0^+ \equiv \{ e_1 \pm e_2, e_2 \pm e_3, e_1 \pm e_3,  e_4 \} , \qquad  \Delta_1^+ \equiv \{ e_1 \pm e_4,  e_2 \pm e_4, e_3 \pm e_4,    e_1, e_2, e_3 \} ,
\end{aligned}\end{equation}
The reduction is given by the Weyl-group element $w_4 = S_{e_1-e_2} S_{e_2-e_3} S_{e_4}$, i.e. $ w_4^6 =\id$.
Obviously this Weyl group element leaves invariant $\Delta_0$ and $\Delta_1$
and the roots $\Delta_1$ are split into four orbits:
 \begin{equation}\label{eq:Or1}\begin{aligned}
\mathcal{O}_1^\pm \quad &\colon \quad \pm e_1-e_4 \to   \pm e_2 +e_4 \to  \pm e_3 - e_4 \to  \pm e_1 + e_4 \to  \pm e_2-e_4   \to  \pm e_3 +e_4  , \\
\mathcal{O}_3 \quad &\colon \quad e_1 \to    e_2 \to  e_3,  \qquad \mathcal{O}_4 \quad \colon \quad -e_1 \to   - e_2 \to  -e_3.
 \end{aligned}\end{equation}
So we have four orbits. After applying another $\mathbb{Z}_2$-reductions, i.e. $Q=Q^\dag$ one  may expect a 2-component NLS.
Realization of the automorphism $w_4$ is as follows $w_4(X) = A_1 X  A_1^{-1}$:
\begin{equation}\label{eq:w4}\begin{split}
A_1 = \left(\begin{array}{ccc} \a_1 & 0 & 0 \\ 0 & \a_2 & 0 \\ 0 & 0 & \a_3 \end{array}\right) , \quad
\a_1 = \left(\begin{array}{ccc} 0 & 1 & 0 \\ 0 & 0 & 1 \\ - 1 & 0 & 0  \end{array}\right) , \quad
\a_2 = \left(\begin{array}{ccc} 0 & 0 & 1 \\ 0 & -1 & 0 \\ 1 & 0 & 0  \end{array}\right), \quad
\a_3 = \left(\begin{array}{ccc} 0 & 0 & -1 \\ -1 & 0 & 0 \\ 0 & -1 & 0  \end{array}\right).
\end{split}\end{equation}
The orthogonality condition is given by (\ref{eq:ort}) with
$S_0 = \left(\begin{array}{ccc} 0 & 0 & s_1 \\ 0 & -s_1 & 0 \\ s_1 & 0 & 0  \end{array}\right) $, $
s_1 = \left(\begin{array}{ccc} 0 & 0 & 1 \\ 0 & -1 & 0 \\ 1 & 0 & 0  \end{array}\right) $
The corresponding potential of the Lax operator is as in (\ref{eq:UV4}) with
\begin{equation}\label{eq:qp-red}\begin{split}
  \q(x,t)  = \left(\begin{array}{ccc} q_1 & q_2 & -q_1 \\  -q_1 & -q_2 & q_1 \\  q_1 & q_2 & -q_1   \end{array}\right), \quad
 \p(x,t)  = \left(\begin{array}{ccc} p_1 & -p_1 & p_1 \\  p_2 & -p_2 & p_2 \\ - p_1 & p_1 & -p_1   \end{array}\right) .
\end{split}\end{equation}
The condition $Q=Q^\dag$ reduces to $p_1=q_1^*$ and $p_2=q_2^*$. Then the NLEE becomes
\begin{equation}\label{eq:nlee1}\begin{split}
i \frac{\partial q_1}{ \partial t } &+ \frac{\partial ^2 q_1}{ \partial x^2 } + 6 (|q_1|^2 + |q_2|^2)q_1(x,t) -3 q_2 ^2 q_1^*(x,t) =0, \\
i \frac{\partial q_2}{ \partial t } &+ \frac{\partial ^2 q_2}{ \partial x^2 } + 3 ( 4|q_1|^2 + |q_2|^2)q_2(x,t) -6 q_1 ^2 q_2^*(x,t)=0,
\end{split}\end{equation}
The corresponding Hamiltonian is:
\begin{equation}\label{eq:H}\begin{split}
H = 2 \left| \frac{\partial q_1}{ \partial x } \right|^2 + \left| \frac{\partial q_2}{ \partial x } \right|^2  - \frac{3}{2} \left( 2|q_1|^2 + |q_2|^2\right)^2
+ 3 \left( q_1 q_2^* - q_1^*q_2\right)^2.
\end{split}\end{equation}
The change of variables:  $v_1= \sqrt{6}q_1$,  $v_2=\sqrt{3}q_2$ leads to:
\begin{equation}\label{eq:NLS2}\begin{split}
i \frac{\partial v_1}{ \partial t } &+ \frac{\partial ^2 v_1}{ \partial x^2 } + (|v_1|^2 + 2|v_2|^2)v_1(x,t) - v_2 ^2 v_1^*(x,t) =0, \\
i \frac{\partial v_2}{ \partial t } &+ \frac{\partial ^2 v_2}{ \partial x^2 } +  ( 2|v_1|^2 + |v_2|^2)v_2(x,t) - v_1 ^2 v_2^*(x,t)=0, \\
\q(x,t)  &= \frac{1}{\sqrt{6}}\left(\begin{array}{ccc} v_1 & \sqrt{2}v_2 & -v_1 \\  -v_1 & - \sqrt{2}v_2 & v_1 \\  v_1 &  \sqrt{2}v_2 & -v_1   \end{array}\right),
\qquad \p = \q^\dag .
\end{split}\end{equation}
 It is easy to check that assuming canonical Poisson brackets  $\{ v_k(x), v_j^*(y)\} = \delta_{kj} \delta(x-y) $ for $v_j$,  the canonical
 Hamiltonian equations of motion:
 \begin{equation}\label{eq:Hvj}\begin{split}
  i \frac{\partial v_j}{ \partial t} = \frac{\delta H' }{\delta v_j^*} , \qquad j=1,2,
 \end{split}\end{equation}
 with
 \begin{equation}\label{eq:H'}\begin{split}
 H'= \left| \frac{\partial v_1}{ \partial x } \right|^2 + \left| \frac{\partial v_2}{ \partial x } \right|^2  - \frac{1}{2} \left( |v_1|^2 + |v_2|^2\right)^2
+ \frac{1}{2} \left( v_1 v_2^* - v_1^*v_2\right)^2,
 \end{split}\end{equation}
 coincides with (\ref{eq:NLS2}).

\section{Spectral properties of Lax operators}

\subsection{The case $SO(2r+1)/(SO(2r-1)\times SO(2))$}

Here we will outline the methods of solving the  direct and the inverse scattering problem (ISP) for $L$.
We will use  the Jost solutions  which are defined by, see \cite{VSG2} and the references therein
\begin{equation}
\lim_{x \to -\infty} \phi(x,t,\lambda) e^{  i \lambda^k J x }=\openone, \qquad  \lim_{x \to \infty}\psi(x,t,\lambda) e^{  i \lambda^k J x } = \openone
 \end{equation}
and the scattering matrix $T(\lambda,t)\equiv \psi^{-1}\phi(x,t,\lambda)$. Due to the special choice of $J$ and
to the fact that the Jost solutions and the scattering matrix take values in the group $SO(2r+1)$ we can use the following
block-matrix structure of $T(\lambda,t)$ and its inverse $\hat{T}(\lambda,t)$:
\begin{equation}\label{eq:25.1}
T(\lambda,t) = \left( \begin{array}{ccc} m_1^+ & -\vec{B}^-{}^T & c_1^- \\
\vec{b}^+ & {\bf T}_{22} & -s_0\vec{b}^- \\ c_1^+ & \vec{B}^+{}^Ts_0 & m_1^- \\ \end{array}\right), \qquad
\hat{T}(\lambda,t) = \left( \begin{array}{ccc} m_1^- & \vec{b}^-{}^T & c_1^- \\
-\vec{B}^+ & \hat{\bf T}_{22} & s_0\vec{B}^- \\ c_1^+ & -\vec{b}^+{}^Ts_0 & m_1^+ \\ \end{array}\right),
\end{equation}
where $\vec{b}^\pm (\lambda,t)$ and $\vec{B}^\pm (\lambda,t)$ are $2r-1$-component vectors,
${\bf T}_{22}(\lambda)$ and $\hat{\bf T}_{22}(\lambda)$ are  $2r-1 \times 2r-1$ blocks and $m_1^\pm
(\lambda)$,  $c_1^\pm (\lambda)$ are scalar functions satisfying
$c_1^\pm = 1/2 (\vec{b}^\pm \cdot s_0 \vec{b}^\pm) /m_1^\mp$.

Important tools for reducing the ISP to a Riemann-Hilbert problem (RHP) are the fundamental analytic solution (FAS) $\chi^{\pm}
(x,t,\lambda )$. Their construction is based on the generalized Gauss decomposition of $T(\lambda,t)$
\begin{equation}\label{eq:FAS_J}
\chi ^\pm(x,t,\lambda)= \phi (x,t,\lambda) S_{J}^{\pm}(t,\lambda ) = \psi (x,t,\lambda ) T_{J}^{\mp}(t,\lambda ) D_J^\pm (\lambda).
\end{equation}
Here $S_{J}^{\pm} $ and $T_{J}^{\pm} $ are upper- and lower-block-triangular matrices, while $D_J^\pm(\lambda)$ are
block-diagonal matrices with the same block structure as $T(\lambda,t)$ above. Skipping the details we give the explicit
expressions of the Gauss factors in terms of the matrix elements of $T(\lambda,t)$
\begin{equation}\label{eq:S_Jpm}\begin{aligned}
S_J^\pm (t,\lambda )&= \exp \left( \pm \sum_{\beta  \in\Delta_1^+ }^{}\tau^\pm_\beta  (\lambda,t) E_{\pm\beta }  \right), &\quad
T_J^\pm (t,\lambda ) &= \exp \left( \mp \sum_{\beta  \in\Delta_1^+ }^{}\rho^\pm_\beta  (\lambda,t) E_{\pm\beta } \right),  \\
D_J^+ &= \left( \begin{array}{ccc} m_1^+ & 0 & 0 \\ 0 & {\bf m}_2^+ & 0 \\
0 & 0 & 1/m_1^+ \end{array} \right), &\quad  D_J^- &= \left( \begin{array}{ccc} 1/m_1^- & 0 & 0 \\ 0 & {\bf m}_2^- & 0 \\
0 & 0 & m_1^- \end{array} \right),
\end{aligned}\end{equation}
where $\vec{b}^+ = (T_{1,2}, \dots , T_{1,2r})^T$
\begin{equation}\label{eq:rotau}\begin{aligned}
 \vec{\tau}^+ (\lambda ,t) &= \frac{\vec{b}^- }{m_1^+} , &\;   \vec{\tau}^- (\lambda ,t) &=  \frac{ \vec{B}^+}{m_1^-} , &\;  \m_2^+ &= \T_{22}+ \frac{\vec{b}^+ \vec{b}^{-,T} }{2m_1^+}
 =\hat{ \T}_{22}+ \frac{s_0\vec{b}^- \vec{b}^{+,T}s_0 }{2m_1^+} , \\
 \vec{\rho}^+ (\lambda ,t) &= \frac{\vec{b}^+ }{m_1^+} , &\;   \vec{\rho}^- (\lambda ,t) &= \frac{\vec{b}^- }{m_1^-} , &\;
\quad \m_2^- &= \hat{\T}_{22}+ \frac{\vec{B}^+ \vec{B}^{-,T} }{2m_1^-}= \hat{\T}_{22}+ \frac{s_0\vec{B}^- \vec{B}^{+,T}s_0 }{2m_1^-} , .
\end{aligned}\end{equation}

If $Q(x,t) $ evolves according to (\ref{eq:eq1}) then the scattering matrix and its elements satisfy the following linear evolution equations
\begin{equation}\label{eq:evol}
i\frac{d\vec{B}^{\pm}}{d t} \pm \lambda ^{2k} \vec{B}^{\pm}(t,\lambda ) =0,  \qquad
i\frac{d\vec{b}^{\pm}}{d t} \pm \lambda ^{2k} \vec{b}^{\pm}(t,\lambda ) =0,  \qquad  i\frac{d m_1^{\pm}}{d t}  =0,
 \qquad  i \frac{d{\bf m}_2^{\pm}}{d t}  =0,
\end{equation}
so the block-diagonal matrices $D^{\pm}(\lambda)$ are  generating functionals of the integrals of motion.
The fact that all $(2r-1)^2$ matrix elements of $m_2^\pm(\lambda)$ for $\lambda \in \bbbc_\pm$  generate integrals of motion reflects
the super-integrability of the model and is due to the degeneracy of the dispersion law determined by $\lambda^{2k} J$. We remind that
$D^\pm_J(\lambda)$ allow analytic extension for $\lambda\in \bbbc_\pm$ and that their zeroes and
 poles determine the discrete eigenvalues of $L$. We will use also another set of FAS:
\begin{equation}\label{eq:chi'}\begin{split}
\chi^{\prime,\pm} (x,t,\lambda) = \chi^\pm (x,t,\lambda) \hat{D}^\pm_J (\lambda).
\end{split}\end{equation}

The FAS for real $\lambda^k$ are linearly related
\begin{equation}\label{eq:rhp0}\begin{split}
\chi^+(x,t,\lambda) &=\chi^-(x,t,\lambda) G_J(\lambda,t), \qquad G_{0,J}(\lambda,t) =S^-_J(\lambda,t)S^+_J(\lambda,t) , \\
\chi^{\prime,+}(x,t,\lambda) &=\chi^{\prime,-}(x,t,\lambda) G'_J(\lambda,t), \qquad G'_{0,J}(\lambda,t) =T^+_J(\lambda,t)T^-_J(\lambda,t) .
\end{split}\end{equation}
One can rewrite eq. (\ref{eq:rhp0}) in an equivalent form for the FAS
$\xi^\pm(x,t,\lambda)=\chi^\pm (x,t,\lambda)e^{i\lambda Jx }$ and $\xi^{\prime,\pm}(x,t,\lambda)=\chi^{\prime,\pm} (x,t,\lambda)e^{i\lambda Jx }$
which satisfy also the relation
\begin{equation}\label{eq:rh-n}
\lim_{\lambda \to \infty} \xi^\pm(x,t,\lambda) = \openone, \qquad \lim_{\lambda \to \infty} \xi^{\prime,\pm}(x,t,\lambda) = \openone.
\end{equation}
Then for $\im \lambda^k=0$ these FAS satisfy
\begin{equation}\label{eq:rhp1}\begin{aligned}
\xi^+(x,t,\lambda) &=\xi^-(x,t,\lambda) G_J(x,\lambda,t), &\quad G_{J}(x,\lambda,t) &=e^{-i\lambda^k
Jx}G_{0,J}(\lambda,t)e^{i\lambda^k Jx} ,  \\
\xi^{\prime,+}(x,t,\lambda) &=\xi^{\prime, -}(x,t,\lambda) G'_J(x,\lambda,t), &\quad G'_{J}(x,\lambda,t) &=e^{-i\lambda^k
Jx}G'_{0,J}(\lambda,t)e^{i\lambda^k Jx} .
\end{aligned}\end{equation}
Obviously the sewing function $G_j(x,\lambda,t)$ is uniquely determined by
the Gauss factors of $T(\lambda,t)$. In view of eq. (\ref{eq:S_Jpm}) we arrive to the following
\begin{lemma}\label{lem:ms}
Let the potential $Q(x,t)$ be such that the Lax operator $L$ has no discrete eigenvalues. Then as minimal set of scattering data
which determines uniquely the scattering matrix $T(\lambda,t)$ and the corresponding potential $Q(x,t)$ one can consider either one
of the sets $\mathfrak{T}_i$, $i=1,2$
\begin{equation}\label{eq:T_i}
\mathfrak{T}_1 \equiv \{ \vec{\rho}^+(\lambda,t), \vec{\rho}^-(\lambda,t),  \quad \lambda^k \in \bbbr\},
\qquad \mathfrak{T}_2 \equiv \{ \vec{\tau}^+(\lambda,t), \vec{\tau}^-(\lambda,t),  \quad \lambda^k \in \bbbr\}.
\end{equation}\end{lemma}

\begin{proof} i) From the fact that $T(\lambda,t)\in SO(2r+1)$ one can derive that
\begin{equation}\label{eq:25.3}
\frac{1}{m_1^+m_1^-} = 1 + (\vec{\rho^+},\vec{\rho^-}) + \frac{1}{4}
(\vec{\rho^+},s_0\vec{\rho^+}) (\vec{\rho^-}, s_0\vec{\rho^-})
\end{equation}
for $\lambda\in\bbbr$. Using the analyticity properties of $m_1^\pm $  we can recover them from eq. (\ref{eq:25.3})
using Cauchy-Plemelji formulae. Given $\mathfrak{T}_i$ and $m_1^\pm$ one easily recovers $\vec{b}^\pm(\lambda)$
and $c_1^\pm(\lambda)$. In order to recover ${\bf m}_2^\pm$ one again uses their analyticity properties,
 only now the problem reduces to a RHP for functions on $SO(2r+1)$. The details will be presented elsewhere.

ii) Given $\mathfrak{T}_i$ one uniquely recovers the sewing function $G_J(x,t,\lambda)$.
In order to recover the corresponding potential $Q(x,t)$ one can use the fact that the RHP (\ref{eq:rhp1})
with canonical normalization has unique regular solution $\chi^\pm(x,t,\lambda)$. Given  $\chi^\pm(x,t,\lambda)$ we recovers $Q(x,t)$ via:
\begin{equation}\label{eq:QQ}
 Q(x,t) = \lim_{\lambda\to\infty} \lambda \left( J - \chi^\pm J \widehat{\chi}^\pm(x,t,\lambda)\right) =
 \lim_{\lambda\to\infty} \lambda \left( J - \chi^{\prime,\pm} J \widehat{\chi}^{\prime,\pm}(x,t,\lambda)\right)..
\end{equation}
which is well known.
\end{proof}
We impose also the standard reduction, namely assume that  $Q(x,t)=Q^\dag(x,t)$,
 or in components $p_k=q_k^*$.
As a consequence we have $\vec{\rho}^-(\lambda,t)=\vec{\rho}^{+,*}(\lambda,t)$
 and $\vec{\tau}^-(\lambda,t)=\vec{\tau}^{+,*}(\lambda,t)$.

\subsection{The case $SO(9)/(SO(3)\times SO(6))$}

Effects on the scattering data:
\begin{itemize}
  \item $T(\lambda)$ belongs to $SO(9)$, therefore $T^{-1} = S_0 T^T(\lambda) S_0$;
  \item $T(\lambda)$ is unitary matrix  $T^\dag (\lambda^*) = T^{-1}(\lambda)$;
  \item $T(\lambda)$ is invariant with respect to the automorphism $A_1$
\end{itemize}
We  parametrize $T(t,\lambda)$ using the same block-matrix structure as for $Q(x,t)$ and $J$ (\ref{eq:UV4'}):
\begin{equation}\label{eq:T}\begin{split}
T(\lambda) =\left(\begin{array} {ccc} \m^+ & -\b^- & \c^- \\ \b^+ & \T_{22} & -\B^- \\ \c^+ & \B^+ & \m^-  \end{array}\right) , \qquad
T^{-1}(\lambda)=\left(\begin{array} {ccc} \s_1\m^{-,T} \s_1 &  \s_1\B^{-,T}\s_1  & \s_1 \c^{-, T} \s_1 \\
- \s_1\B^{+,T} \s_1 & \s_1 \T_{22}^T \s_1 & \s_1 \b^{-,T} \s_1 \\  \s_1 \c^{+,T} \s_1 & - \s_1 \b^{+,T} \s_1 & \s_1 \m^{-,T} \s_1  \end{array}\right) ,
\end{split}\end{equation}
\begin{equation}\label{eq:T-1}\begin{split}
T^{-1} = S_0 T^T(\lambda) S_0,  \qquad  T^\dag (\lambda^*) = T^{-1}(\lambda), \qquad T (\lambda) = A_1 T(\lambda) A_1^{-1} ,
\end{split}\end{equation}
i.e.
\begin{equation}\label{eq:tt}\begin{aligned}
\m^{+,\dag }(\lambda^*) & = \s_1 \m^{-,T}(\lambda) \s_1, &\quad \c^{+,\dag }(\lambda^*) & = \s_1 \c^{-,T}(\lambda) \s_1, \\
\b^{\pm,\dag }(\lambda^*) & = \s_1 \B^{\mp ,T}(\lambda) \s_1, &\quad \B^{\pm,\dag }(\lambda^*) & = \s_1 \b^{\mp ,T}(\lambda) \s_1 .
\end{aligned}\end{equation}
and
\begin{equation}\label{eq:Tred}\begin{aligned}
\m^+ & = \a_1 \m^+ \a_1^{-1}, &\quad  \b^- & = \a_1 \b^- \a_2^{-1}, &\quad  \c^- & = \a_1 \c^- \a_3^{-1}, \\
\b^+ & = \a_2 \b^+ \a_1^{-1}, &\quad  \T_{22} & = \a_2 \T_{22} \a_2^{-1}, &\quad  \B^- & = \a_2 \B^- \a_3^{-1}, \\
\c^+ & = \a_3 \c^+ \a_1^{-1}, &\quad  \B^+ & = \a_3 \B^+ \a_2^{-1}, &\quad  \m^- & = \a_3 \m^- \a_3^{-1} .
\end{aligned}\end{equation}

%%%%%%%%%%%%%%%%%%%%%%%%%%

\section{Integrability properties of Kulish-Sklyanin models}

\subsection{The Wronskian relations and minimal sets of scattering data}\label{ch:WR}

%\subsection{The Wronskian relations}\label{sec:l5-1}

The analysis of the mapping $\mathcal{ F} \colon \mathcal{ M} \to
\mathcal{ T} $ between the class of allowed potentials $\mathcal{M} $ and the scattering data of $L $ starts with the Wronskian
relations, see \cite{DJK,CalDeg} for $sl(2) $-case and \cite{VSG2,ContMat}. For higher rank algebras and symmetric spaces
(the block-matrix case) one should use \cite{IP2,ContMat}.

These ideas  will be worked out for $L$ (\ref{eq:L}) with $s=1$. With it one can associate:
\begin{eqnarray}\label{eqB:I.3}
&& i {d \hat{\psi} \over d x }- \hat{\psi} (x,t,\lambda )U(x,t,\lambda ) =0, \qquad U(x,\lambda )=Q(x)-\lambda J, \\
\label{eqB:I.4} && i {d \delta \psi \over d x }  + \delta U(x,t,\lambda ) \psi (x,t,\lambda ) + U(x,t,\lambda ) \delta \psi (x,t,\lambda ) =0\\
\label{eqB:I.5} && i {d \dot{\psi} \over d x }  - \lambda J \psi (x,t,\lambda ) + U(x,t,\lambda ) \dot{\psi} (x,t,\lambda ) =0
\end{eqnarray}
where $\delta \psi  $ corresponds to a given variation $\delta Q(x,t) $ of the potential, while by dot we denote the derivative
with respect to the spectral parameter. We start with the identity:
\begin{eqnarray}\label{eqB:wr.1}
\left. \left( \hat{\chi }J \chi (x,\lambda ) - J \right) \right|_{x=-\infty }^{\infty } &=& i \int_{-\infty }^{\infty } d
x\, \hat{\chi }[J,Q(x)]\chi (x,\lambda ),
\end{eqnarray}
where $\chi (x,\lambda ) $ can be any fundamental solution of $L$. For convenience we use both
choices  $\chi (x,\lambda ) =\chi^\pm (x,\lambda ) $   and $\chi (x,\lambda ) =\chi^{\prime,\pm} (x,\lambda ) $.

The left hand side of (\ref{eqB:wr.1}) can be calculated explicitly by using the asymptotics of $\chi ^\pm(x,\lambda ) $
for $x\to \pm \infty  $, (\ref{eq:FAS_J}).  It would be expressed by the matrix elements of the scattering matrix $T(\lambda ) $,
i.e., by the scattering data of $L $.

Indeed, let us  multiply both sides of eq. (\ref{eqB:wr.1}) by  $E_{\pm \beta} $, $\beta\in \Delta_{1}^{+}$ and take the Killing form.
In the right hand side of this equation we can use the invariance properties of the trace and rewrite it in the form:
\begin{equation}\label{eqB:wr.1ab}\begin{split}
\langle \left. \left( \hat{\chi }^\pm J \chi^\pm  (x,\lambda ) - J \right)E_{\beta}\rangle\right|_{x=-\infty }^{\infty } &= i
\int_{-\infty }^{\infty } d x\, \langle \left( [J, Q(x)] \e_{\beta}^\pm (x,\lambda )\right) \rangle ,\\
\langle \left. \left( \hat{\chi }^{\prime,\pm} J \chi^{\prime,\pm} (x,\lambda ) - J\right)E_{\beta}\rangle\right|_{x=-\infty }^{\infty }&= i
\int_{-\infty }^{\infty } d x\, \langle \left( [J, Q(x)] \e_{\beta}^{\prime,\pm} (x,\lambda )\right) \rangle ,
\end{split}\end{equation}
where
\begin{equation}\label{eq:e-ab}\begin{aligned}
e_{\beta}^\pm (x,\lambda )& =\chi^\pm E_{\beta}\hat{\chi }^\pm (x,\lambda ) , &\quad \e_{\beta}^\pm (x,\lambda ) &=\pi_{0J}(\chi^\pm
E_{\beta}\hat{\chi }^\pm  (x,\lambda ))  ,\\
e_{\beta}^{\prime,\pm}  (x,\lambda )& =\chi^{\prime,\pm} E_{\beta}\hat{\chi }^{\prime,\pm}  (x,\lambda ) , &\quad
\e_{\beta}^{\prime,\pm} (x,\lambda ) &=\pi_{0J}(\chi^{\prime,\pm} E_{\beta}\hat{\chi }^{\prime,\pm}  (x,\lambda ))  ,
\end{aligned}\end{equation}
are the natural generalization of the `squared solutions' introduced first for the $sl(2) $-case by Kaup  \cite{DJK} and
generalized to any simple Lie algebra in \cite{IP2,VSG2,ContMat}. By
$P_{0J}$ we have denoted the projector $\pi_{0J}=\ad_J^{-1}\ad_J$ on the block-off-diagonal part of the corresponding matrix-valued
function.

The right hand sides of eq. (\ref{eq:e-ab}) can be written down with the skew--scalar product:
\begin{equation}\label{skew}
\biglb X, Y \bigrb = \int_{-\infty}^\infty d x \langle  X(x), [J, Y(x) ] \rangle ,
\end{equation}
where $\langle X,Y\rangle$ is the Killing form; in what follows we assume that the Cartan-Weyl generators satisfy $\langle E_\alpha,
E_{-\beta} \rangle =\delta_{\alpha,\beta}$ and $\langle H_j, H_k \rangle =\delta_{jk}$. The  product is skew-symmetric $\biglb X, Y
\bigrb = - \biglb Y, X \bigrb$ and is non-degenerate on the space of allowed potentials $ \mathcal{M} $. Thus we find
\begin{equation}\label{eq:rho-ea} \begin{split}
\rho_{\beta}^{\pm}=-i \biglb Q(x),\e^{\prime,\pm}_{\pm \beta} \bigrb, \qquad \tau_{\beta}^{\pm}=-i \biglb Q(x),\e^{\pm}_{\mp\beta} \bigrb,
\end{split}\end{equation}

Thus the mappings  $\mathfrak{F}:Q(x,t) \to \mathfrak{T}_i$ can be viewed as generalized Fourier transform in which
$\e_{\beta}^{\pm}(x,\lambda ) $ and $\e_{\beta}^{\prime,\pm}(x,\lambda ) $ can be viewed as
generalizations of the standard exponentials. In what follows we will show that the same `squared solutions' appear in the analysis
of the mapping between the variations $\delta Q(x,t)$ and $\delta \mathfrak{T}_i$.

The second type of Wronskian relations  relate the variation of the potential $\delta Q(x) $ to the
corresponding variations of the scattering data. To this purpose we start with the identity:
\begin{equation}\label{}
\left.  \hat{\chi }^\pm \delta \chi^\pm (x,\lambda ) \right|_{x=-\infty }^{\infty } = i \int_{-\infty }^{\infty } d x\,
\hat{\chi } \delta Q(x)\chi (x,\lambda ),
\end{equation}
which follows from eqs. (\ref{eqB:I.3}) and (\ref{eqB:I.5}{{61}). We apply ideas similar to the ones above and get:
\begin{equation}\label{eq:drho-ea}\begin{aligned}
\delta\rho_{\beta}^{\pm}=\mp i \biglb \ad_J^{-1}\delta Q(x),\e^{\prime,\pm}_{\pm\beta} \bigrb, \qquad
\delta\tau_{\beta}^{\pm}= \pm i \biglb \ad_J^{-1}\delta Q(x),\e^{\pm}_{\mp \beta} \bigrb,
\end{aligned}\end{equation}
where $\beta \in \Delta_1^+$.
These relations are basic in the analysis of the related NLEE and their Hamiltonian structures. Below we shall use them assuming
that the variation of $Q(x) $ is due to its time evolution, and consider variations of the type:
\begin{equation}\label{eqB:wr.23}
\delta  Q(x,t) = Q_t \delta t + \mathcal{ O} ((\delta  t)^2).
\end{equation}
Keeping only the first order terms with respect to $\delta t $ we find:
\begin{equation}\label{eq:drho-ea'}\begin{aligned}
\frac{d\rho_{\beta}^{\pm}}{dt}= \mp i \biglb \ad_J^{-1} Q_t(x),\e^{\prime,\pm}_{\pm \beta} \bigrb, &\quad
\frac{d\tau_{\beta}^{\pm}}{dt}= \pm i \biglb \ad_J^{-1} Q_t(x),\e^{\pm}_{\mp \beta} \bigrb.
\end{aligned}\end{equation}

%%%%%%%%%%%%%%%%%%%%%%%%%%%%%

%\subsection{Expansions over squared solutions and Hamiltonian hierarchies}
\subsection{The generalized Fourier transforms and the completeness of the `squared solutions' }\label{ch:GFT}

It is known that the `squared solutions' $\e_{\alpha} ^\pm (x, \lambda ) = \pi_{0J} \left(\chi^\pm E_\alpha \chi^\pm(x,t,\lambda)\right) $,
form complete set of functions in the space of allowed potentials $q(x)$, see
\cite{VSG2,ContMat}.
For brevity and simplicity below we assume that $L$ has no discrete eigenvalues.
Let us introduce the sets of `squared solutions'
\begin{equation}\label{eq:Phi}\begin{split}
\{\bPsi \} &\equiv \left\{ \e ^+_{-\alpha}(x,\lambda), \quad \e ^-_{\alpha}(x,\lambda), \quad \lambda \in \bbbr
, \quad \alpha\in\Delta_1^+ \right\}, \\
\{\bPhi \} &\equiv \left\{ \e ^+_{\alpha}(x,\lambda), \quad \e ^-_{-\alpha}(x,\lambda), \quad \lambda \in \bbbr
, \quad \alpha\in\Delta_1^+ \right\}.
\end{split}\end{equation}

`\begin{theorem}[see \cite{VSG2,ContMat}]\label{t2.1}
The sets $\{\bPsi \} $  and $\{\bPhi \} $ form complete sets of functions in $\mathcal{M}_J$. The corresponding completeness
relation has the form:
\begin{equation}\label{eq:5.23}
\begin{split}
\delta(x-y)\Pi_{0J} &= {1\over \pi} \int_{-\infty}^\infty d \lambda (G_1^+(x,y,\lambda) - G_1^-(x,y,\lambda) ),
\end{split}\end{equation}
where
\begin{equation}\label{eq:5.23'}
\begin{split}
\Pi_{0J} &=\sum_{\alpha \in \Delta_{1}^{+}}( E_{\alpha}\otimes E_{-\alpha} - E_{-\alpha}\otimes E_{\alpha}) , \\
G_1^\pm (x,y,\lambda) &= \sum_{\alpha \in \Delta_{1}^{+}} \e_{\pm \alpha}^\pm (x,\lambda)\otimes \e_{\mp \alpha}^+(y,\lambda),
\end{split}\end{equation}
\end{theorem}

\begin{proof}[Idea of the proof]
Apply the contour integration method to the Green function
\begin{equation}\label{eq:G}\begin{split}
G^\pm(x,y,\lambda) &= G_1^\pm(x,y,\lambda) \theta(y-x) - G_2^\pm(x,y,\lambda) \theta(x-y), \\
G_2^\pm (x,y,\lambda) &= \sum_{\alpha \in \Delta_0\cup \Delta_{1}^{-}} \e_{\pm\alpha}^-(x,\lambda)\otimes
\e_{\mp\alpha}^-(y,\lambda) + \sum_{j=1}^r \h_{j}^\pm (x,\lambda)\otimes \h_{j}^\pm(y,\lambda),
\end{split}\end{equation}
where $\h_{j}^\pm (x,\lambda) = \pi_{0J} \left( \chi^\pm(x,\lambda) H_j \hat{\chi}^\pm (x,\lambda) \right),$ and calculate  the integral
\begin{equation}\label{eq:CIM}
\mathcal{J}_G(x,y) = \frac{1}{2\pi i} \oint_{\gamma_+} d\lambda \; G^+(x,y,\lambda) - \frac{1}{2\pi i} \oint_{\gamma_-} d\lambda \;
G^-(x,y,\lambda),
\end{equation}
in two ways: i) via the Cauchy residue theorem and ii) integrating along the contours, see \cite{VSG2,ContMat}.

\end{proof}

  Skipping the details we write down the expansions of $q(x) $ and $\ad_J^{-1}\delta q(x)$ assuming $L$ has no discrete spectrum:
\begin{equation}\label{eq:49.4}
Q(x) = -{i\over \pi } \int_{-\infty }^{\infty } d \lambda \sum_{\alpha\in\Delta_1^+} \left( \tau^+_{\alpha}(\lambda )
\e_{\alpha} ^+(x, \lambda ) -\tau_{\alpha}^-(\lambda ) \e_{-\alpha} ^-(x, \lambda ) \right) ,
\end{equation}
\begin{equation}\label{eq:50.6}
\ad_J^{-1}\delta Q(x) = {i \over \pi } \int_{-\infty }^{\infty } d \lambda \sum_{\alpha\in\Delta_1^+} \left(
\delta\tau^+_{\alpha}(\lambda ) \e_{\alpha} ^+(x, \lambda ) + \delta \tau_{\alpha}^-(\lambda )  \e_{-\alpha} ^-(x, \lambda )\right) .
\end{equation}
These expansions can be viewed as tool  to establish the one-to-one correspondence between $q(x) $ (resp. $\ad_J^{-1}\delta q$ and each of the minimal
sets of scattering data $\mathcal{T}_i $ (resp. $\delta\mathcal{T}_i $), $i=1,2$.
To complete the analogy between the standard Fourier transform and the expansions over the `squared solutions' we need the
generating operators $\Lambda _\pm $:
\begin{equation}\label{eq:**6}
\Lambda _\pm X(x) \equiv \ad_{J}^{-1} \left( i {d X \over d x} + i \left[ q(x), \int_{\pm\infty }^{x} d y\, [q(y), X(y)]\right] \right).
\end{equation}
for which the `squared solutions' are eigenfunctions:
\begin{equation}\label{eq:**0}\begin{split}
(\Lambda _+-\lambda )\e_{-\alpha}^{+} (x,\lambda ) = 0, \qquad (\Lambda _+-\lambda )\e_{\alpha}^{-} (x,\lambda ) = 0, \\
(\Lambda _--\lambda )\e_{\alpha}^{+} (x,\lambda ) = 0, \qquad (\Lambda _--\lambda )\e_{-\alpha}^{-} (x,\lambda ) = 0.
\end{split}\end{equation}

\subsection{Fundamental properties of the KS type equations}

The expansions (\ref{eq:49.4}), (\ref{eq:50.6}) and the explicit form of $\Lambda_\pm$ and eq. (\ref{eq:**0}) are basic for deriving the fundamental
properties of all MNLS type equations related to the Lax operator $L$. Each of these NLEE is determined by its dispersion law which we choose
to be of the form $F(\lambda) =f(\lambda) J$, where $f(\lambda)$ is polynomial in $\lambda$. The corresponding NLEE becomes:
\begin{equation}\label{nlee}
i\ad_J^{-1} \frac{\partial Q}{ \partial t } + f(\Lambda_\pm) Q(x,t)  = 0.
\end{equation}

\begin{theorem}\label{t1}
The NLEE \eqref{nlee} are equivalent to: i) the equations (\ref{eq:evol}) and ii) the following evolution equations for
the generalized Gauss factors of $T(\lambda)$:
\begin{equation}\label{ds}
i {dS^+_J \over dt} + [F(\lambda), S^+_J] =  0, \qquad i {dT^-_J \over dt} + [F(\lambda), T^-_J] =  0, \qquad {dD^+_J \over dt}=0.
\end{equation}
\end{theorem}

The principal series of integrals is generated by the asymptotic expansion of $\ln m_1^+(\lambda)= \sum_{k=1}^\infty I_k \lambda^{-k}$.
The first three integrals of motion:
\begin{equation}\begin{split}\label{eq:I1-3}
I_1 = -\frac{i}{2} \int_{-\infty }^{\infty } d x\, \langle Q(x), Q(x) \rangle, \qquad
I_2  = \frac{1}{2} \int_{-\infty }^{\infty } d x\, \langle Q_x, \ad_J^{-1}Q(x) \rangle , \\
I_3 = -\frac{i}{2} \int_{-\infty }^{\infty } d x\, \left ( \left \langle \ad_J^{-1} Q_x, \ad_J^{-1}Q_x \right \rangle
- \left  \langle \left [ \ad_J^{-1} Q, Q(x).\right ], \left [ \ad_J^{-1}Q, Q(x)\right ] \right  \rangle \right ).
\end{split}\end{equation}
Now $iI_1$ can be interpreted as
density of the particles, $I_2$ is the momentum and $I_3=2i H_{\rm MNLS}$. Indeed,  the
Hamiltonian equations of motion provided by $H_{(0)}=-iI_3/2$ with the Poissson brackets
\begin{equation}\label{eq:PB}
\{ q_{\alpha}(y,t) , p_{\beta}(x,t)  \} = i \delta _{\alpha,\beta} \delta (x-y), \qquad \alpha, \beta \in \Delta_1^+,
\end{equation}
coincide with the MNLS equations (\ref{eq:KSm1}). The above Poisson brackets are dual to the canonical symplectic form:
\[ \Omega _0= i \int_{-\infty }^{\infty }d x\, \tr \left(\delta \vec{p}(x) \wedgecomma \delta \vec{q}(x) \right)=\frac{1}{2i}
\biglb \ad_{J}^{-1} \delta q(x) \wedgecomma \ad_{J}^{-1} \delta q(x) \bigrb, \]
where $\wedgecomma $ means that taking the scalar or matrix product we
exchange the usual product of the matrix elements by wedge-product.
The Hamiltonian formulation of eq. (\ref{eq:KSm1}) with $\Omega _0 $ and
$H_0 $ is just one member of the hierarchy of Hamiltonian formulations provided by:
\begin{equation}\label{eq:5.2.6}
\Omega _k = {1 \over i }\biglb \ad_{J}^{-1} \delta Q \wedgecomma
\Lambda ^k \ad_{J}^{-1} \delta Q \bigrb , \qquad  H_k = i^{k+3} I_{k+3}.
\end{equation}
where $\Lambda ={1 \over 2 } (\Lambda_+ +\Lambda _-)$. We can also
calculate $\Omega _k $ in terms of the scattering data variations. Imposing the reduction $q(x)=q^\dag (x)$ we get:
\begin{equation*}
\begin{split} \Omega _k &=  {1  \over 2\pi i } \int_{-\infty }^{\infty } d\lambda \, \lambda ^k \left( \Omega
_{0}^{+}(\lambda ) -\Omega _{0}^{-}(\lambda ) \right)\\ &=  {1  \over 2\pi} \int_{-\infty }^{\infty } d\lambda \, \lambda ^k \im \left( m_1^+ (\lambda)
\left( \hat{\m}_2^+ \delta \vec{\rho}^+(\lambda) \wedgecomma \delta \vec{\tau}^+(\lambda)\right) \right).
\end{split}\end{equation*}
This allows one to prove that if we are able to cast $\Omega_{0} $
in canonical form,  then all $\Omega _k $ will also be cast in canonical form and will be pair-wise equivalent.

\section{Conclusions}

Using RHP formulated on the real axis of the complex $\lambda$-plane and compatible with the {\bf BD.I}-type symmetric spaces
$SO(2r+1)/S(O(2r -2s +1)\otimes O(2s))$, $s\geq 1$ we have derived Lax pairs for KS type models; the proper KS model is obtained for $s=1$.

Another  Riemann-Hilbert problems:  formulated on  $\mathbb{R} \oplus i\mathbb{R}$  is relevant for a new type of KS model.
We find nontrivial deep reductions of these systems and formulate their effects on the scattering matrix.
In particular we obtain new 2-component NLS equations whose Hamiltonian depends not only on $|q_1|$. and $|q_2|$, but also on
$q_1^*q_2 + q_1 q_2^*$. Thus our example comes out of the scope of Zakharov-Schulman theorem \cite{ZahSch}.

Finally, using the Wronskian relations we demonstrate that the inverse scattering method for
KS models may be understood as a generalized Fourier transforms. Thus we have a tool to derive all their fundamental properties, including the hierarchy of
equations and he hierarchy of  their Hamiltonian structures.

\section*{Acknowledgement}
I am grateful to professors F. Calogero and V. E. Zakharov for useful suggestions and comments,
and to an anonymous referee for careful reading of the manuscript.

\end{document}